\documentclass[journal,onecolumn]{IEEEtran}

\usepackage{amsmath, amsthm, amssymb}
\usepackage{latexsym}
\usepackage{graphicx}
\usepackage{verbatim}
\usepackage{mathrsfs}
\usepackage{float}
\usepackage[lofdepth, lotdepth]{subfig}
\usepackage{array}
\usepackage{url}
\usepackage{algorithm}
\usepackage[noend]{algorithmic}
\usepackage{cite}
\usepackage[table]{xcolor}
\usepackage{enumerate}
\usepackage{eucal}
\usepackage{stmaryrd}
\usepackage{mathtools}
\usepackage{pgf}
\usepackage{tikz}
\usetikzlibrary{arrows,automata}
\usepackage[latin1]{inputenc}
\usetikzlibrary{automata,positioning}

\usepackage[top=0.6in, bottom=0.6in, left=0.55in, right=0.55in]{geometry}
\usepackage{mathtools}
\DeclarePairedDelimiter\ceil{\lceil}{\rceil}
\DeclarePairedDelimiter\floor{\lfloor}{\rfloor}

\theoremstyle{plain}
\newtheorem{theorem}{Theorem}
\newtheorem{proposition}{Proposition}
\newtheorem{lemma}{Lemma}

\newtheorem{corollary}{Corollary}

\theoremstyle{definition}
\newtheorem{definition}{Definition}
\newtheorem{example}{Example}

\DeclareMathAlphabet{\mathbfsl}{OT1}{ppl}{b}{it} 

\newcommand{\bp}{{\mathbfsl p}}
\newcommand{\bq}{{\mathbfsl q}}

\newcommand{\bu}{{\mathbfsl u}}
\newcommand{\bv}{{\mathbfsl v}}

\newcommand{\by}{{\mathbfsl y}}

\newcommand{\bc}{{\mathbfsl c}}

\newcommand{\bw}{{\mathbfsl{w}}}

\newcommand{\bx}{{\mathbfsl{x}}}
\newcommand{\bz}{{\mathbfsl{z}}}

\newcommand{\bbZ}{{\mathbb Z}}

\newcommand{\ppmod}[1]{~({\rm mod~}#1)}

\renewcommand{\ge}{\geqslant}
\renewcommand{\le}{\leqslant}

\newcommand{\et}{{\emph{et al.}}}

\newcommand{\enc}{\textsc{Enc}}
\newcommand{\dec}{\textsc{Dec}}

\IEEEoverridecommandlockouts
\begin{document}

\pagestyle{empty}

\title{Efficient Design of Subblock Energy-Constrained Codes and Sliding Window-Constrained Codes}
\author{
   \IEEEauthorblockN{
   	Tuan Thanh Nguyen,
	Kui Cai, 
	and Kees A. Schouhamer Immink}
\thanks{Tuan Thanh Nguyen and Kui Cai are with the Singapore University of Technology and Design, Singapore 487372 (email: \{tuanthanh\_nguyen, cai\_kui\}@sutd.edu.sg).}
\thanks{Kees A. Schouhamer Immink is with the Turing Machines Inc, Willemskade 15d, 3016 DK Rotterdam, The Netherlands (email: immink@turing-machines.com).}
\thanks{This paper was presented in part at the 2020 Proceedings of the IEEE International Symposium on Information Theory \cite{binary2020}.}
}
\maketitle

\hspace{-3mm}\begin{abstract}
The subblock energy-constrained codes (SECCs) and sliding window-constrained codes (SWCCs) have recently attracted attention due to various applications in communication systems such as simultaneous energy and information transfer. In a SECC, each codeword is divided into smaller non-overlapping windows, called subblocks, and every subblock is constrained to carry sufficient energy. In a SWCC, however, the energy constraint is enforced over every window. In this work, we focus on the binary channel, where sufficient energy is achieved theoretically by using relatively high weight codes, and study SECCs and SWCCs under more general constraints, namely bounded SECCs and bounded SWCCs. We propose two methods to construct such codes with low redundancy and linear-time complexity, based on Knuth's balancing technique and sequence replacement technique. These methods can be further extended to construct SECCs and SWCCs. For certain codes parameters, our methods incur only one redundant bit. We also impose the minimum distance constraint for error correction capability of the designed codes, which helps to reduce the error propagation during decoding as well.


\end{abstract}


\section{Introduction}
Constrained coding has been used widely in various communication and storage systems. 
For example, to avoid detection errors due to inter-symbol interference and synchronization errors in magnetic and optical storage, {\em runlength-limited codes} (RLLs) are employed to restrict any run of zeros between consecutive ones \cite{kas:1990,book}.  
Recently, the {\em subblock energy-constrained codes} (SECCs) and {\em sliding window-constrained codes} (SWCCs) have been shown as suitable candidates for providing simultaneous energy and information transfer from a powered transmitter to an energy-harvesting receiver \cite{vars:2016,vars:2008,Zhao:2016,skip:2018,chee:2014,hm2017,kiah:2019,kiah:2018,bar2011,fou2014}. In this scenario, the receiver uses the same received signal both for decoding information and for harvesting energy which is to power the receiver's circuitry. In 2008, Varshney \cite{vars:2008} characterized the tradeoff between reliable communication and delivery of energy at the receiver by using a general capacity-power function, where transmitted sequences were constrained to contain sufficient energy. In this work, we focus on the binary channel, where {\em on-off keying} is employed, and bit 1 (bit 0) denotes the presence (absence) of a high energy signal. As such, sufficient energy is achieved theoretically by using relatively high weight codes. 

Recently, Tandon \et{} \cite{vars:2016} demonstrated that imposing only an energy constraint over the whole transmitted sequence might not be sufficient. It is important to avoid sequences which carry limited energy over long duration, and consequently, preventing energy outage at a receiver having finite energy storage capability. In order to regularize the energy content in the signal, two classes of energy-constrained codes, namely SECCs and SWCCs, were suggested \cite{vars:2016,skip:2018}. Formally, in a binary SECC, each codeword is divided into smaller non-overlapping window, called {\em subblocks}, and every subblock is constrained to have sufficient number of ones. In contrast, a binary SWCC restricts the number of ones over every window of consecutive symbols (see Figure~\ref{fig:comparison}). This approach has been investigated in \cite{bar2011,fou2014,immink2020}. SWCCs have been further studied for other applications of error-correction codes in \cite{gabrys2018,tt2019}. In fact, the subblock energy constraint is weaker than the latter, and even if every subblock in a codeword $\bc$ carries sufficient energy, there might still be a subsequence in $\bc$ that carries limited energy over long duration (see Example~\ref{example1}). In contrast, the sliding-window constraint enables all codewords to carry sufficient energy over every duration, which meets real-time delivery requirements, but also reduces the number of valid codewords and therefore the information capacity. In this work, we provide some bounds for SWCCs and show that if the length of each duration satisfies certain constraints, there exist codes whose rates approach capacity. In such cases, we design an efficient method to construct SWCCs with only one redundant bit.

Furthermore, we study SECCs and SWCCs under more general constraints, namely {\em bounded SECCs} and {\em bounded SWCCs}. The additional constraint restricts the energy in every subblock in SECCs (or every window in SWCCs) to be below a given threshold, consequently preventing energy outage at a receiver having finite energy storage capability (see Figure~\ref{fig:comparison}). Throughout this paper, we propose two methods for constructing bounded SECCs and bounded SWCCs, based on {\em Knuth's balancing technique} and {\em sequence replacement technique}. The methods can be extended to construct SECCs and SWCCs as well. We further combine these codes efficiently with {\em error correction codes} (ECCs), which also helps to reduce error propagation of the designed codes during decoding. Before we present the main results of the paper, we go through certain notations and then highlight the major contributions of this work.


\subsection{Notations}\label{sec:prelim}
Given two binary sequences $\bx=x_1\ldots x_m$ and $\by=y_1\ldots y_n$, the {\em concatenation} of the two sequences is defined by $$\bx \by \triangleq x_1\ldots x_m y_1\ldots y_n.$$ For a binary sequence $\bx$, we use ${\rm wt}(\bx)$ to denote the weight of $\bx$, i.e the number of ones in $\bx$. We use $\overline{\bx}$ to denote the complement of $\bx$. For example, if $\bx=00111$ then ${\rm wt}(\bx)=3$ and $\overline{\bx}=11000$.

Throughout this work, we denote the codeword length by $n$, the subblock (or window) length by $\ell$ where $\ell\le n$. In SECCs, we also require $n=m\ell$ for some positive integer $m$. 

\begin{definition} For $0\le a\le \ell$, we use ${\cal S}(n,\ell,a)$ to denote the set of all codewords with length $n = m\ell$, and the weight in each subblock is at least $a$, and we use ${\cal W}(n,\ell,a)$ to denote the set of all codewords with length $n$ (not necessary a multiple of $\ell$), and the weight in every window of size $\ell$ is at least $a$. 
\end{definition}

\begin{definition} For $0\le a< b \le \ell$, we use ${\cal S}(n,\ell,[a,b])$ to denote the set of all codewords with length $n = m\ell$, and the weight in each subblock is at least $a$ and at most $b$. Similarly, ${\cal W}(n,\ell,[a,b])$ denotes the set of all codewords with length $n$ and the weight in every window of size $\ell$ is at least $a$ and at most $b$. 
\end{definition}

\begin{proposition}
For all $0\le a< b \le \ell$, we have
\begin{enumerate}[(i)]
\item ${\cal W}(n,\ell,a)\subset {\cal S}(n,\ell,a)$, ${\cal W}(n,\ell,[a,b]) \subset {\cal S}(n,\ell,[a,b])$,
\item ${\cal W}(n,\ell,a) \equiv {\cal W}(n,\ell,[a,\ell])$, ${\cal S}(n,\ell,a)\equiv{\cal S}(n,\ell,[a,\ell])$.
\end{enumerate}
\end{proposition}

Given $0<\ell\le n$, $0\le a<b\le \ell,$ the capacity of those constraint channels are defined by: 
\begin{align*}
{\bf c}_{\cal S}{(\ell,a)} &\triangleq \lim_{n \to \infty} 1/n \log |{\cal S}(n,\ell,a)|,       \\
{\bf c}_{\cal S}{(\ell,[a,b])} &\triangleq \lim_{n \to \infty} 1/n \log |{\cal S}(n,\ell,[a,b])|, \\
{\bf c}_{\cal W}{(\ell,a)} &\triangleq \lim_{n \to \infty} 1/n \log |{\cal W}(n,\ell,a)|,  \\
{\bf c}_{\cal W}{(\ell,[a,b])} &\triangleq \lim_{n \to \infty} 1/n \log |{\cal W}(n,\ell,[a,b])|. 
\end{align*}

The capacity ${\bf c}_{\cal W}{(\ell,a)}$ is studied and determined for certain values of $\ell$ and $a$ in our companion paper \cite{immink2020}. A special class of bounded SWCCs, namely {\em locally balanced constraints}, was introduced in \cite{ryan:2020} and the capacity ${\bf c}_{\cal W}{(\ell,[a,b])}$ was also studied when $a=\ell/2-\epsilon, b=\ell/2+\epsilon$ for $\epsilon > 0$. 
In general, to achieve high information capacity, the sufficient values for $a,b$ are $a\le p_1 \ell$ and $b\ge p_2\ell$ for some constants $0\le p_1 < 1/2< p_2\le 1$. In this work, not only are we interested in constructing large codes, we desire efficient encoders that map arbitrary binary messages into these codes.

\begin{definition}
For $0\le a\le \ell \le n$, and $0\le r\le n$, an encoder $\enc:\{0,1\}^{n-r} \to\{0,1\}^n$ is a $(n,\ell,a)$-{\em subblock energy-constrained encoder with $r$ bits of redundancy} 
if $\enc(\bx) \in {\cal S}(n,\ell,a)$ for all $\bx\in\{0,1\}^{n-r}$. The rate of the encoder is computed by $(n-r)/n=1-r/n$. For $0\le a<\ell/2<b\le \ell \le n$, the $(n,\ell,[a,b])$-{\em bounded subblock energy-constrained encoder}, $(n,\ell,a)$-{\em sliding window-constrained encoder}, and $(n,\ell,[a,b])$-{\em bounded sliding window-constrained encoder} are defined similarly. 
\end{definition}

For each constraint, our design objectives include low redundancy (equivalently, high information rate) and low complexity of the encoding/decoding algorithms. In Section II and Section III, for certain codes parameters, the rate of our encoders approaches the channel capacity.

\begin{figure*}[t!]
\begin{center}
\includegraphics[width=16cm]{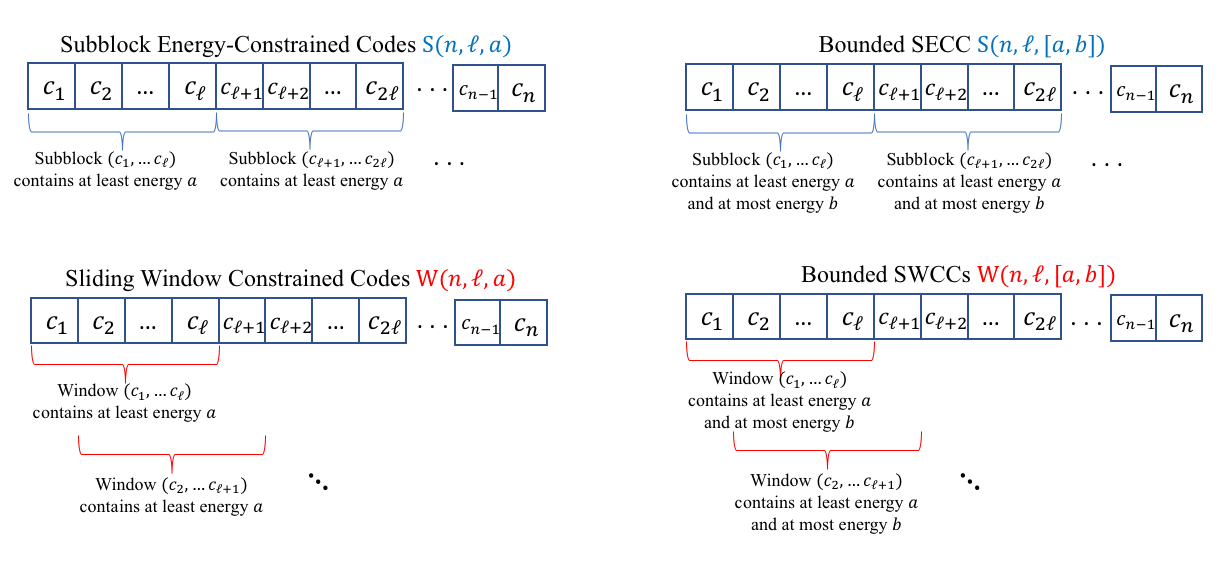}
\end{center}
\caption{SECCs, SWCCs, bounded SECCs, and bounded SWCCs.}
\label{fig:comparison}
\end{figure*}

\begin{definition}
For $n,\ell>0, n=m\ell$, a sequence $\bx=x_1x_2\ldots x_{n}\in\{0,1\}^{n}$ is divided into $m$ subblocks of size $\ell$ where the $i$th subblock is denoted by ${\bf B}_{(i,\ell)}(\bx)$, and ${\bf B}_{(i,\ell)}(\bx)=x_{(i-1)\ell+1}\ldots x_{i\ell}$ for $1\le i\le m$. On the other hand, the $i$th window of size $\ell$ of $\bx$, denoted by $\bw_{(i,\ell)}(\bx)$, is defined by $\bw_{(i,\ell)}(\bx)=x_{i} \ldots x_{i+\ell-1}$ for $1\le i\le n-\ell+1$.
\end{definition} 

\begin{example}\label{example1}
Let $n=18, \ell=6, m=3, a=2, b=5$. Consider a sequence $\bx= 001111 110000 011001$. The subblocks of $\bx$ are defined as follows. 
\begin{equation*}
\bx= \underbrace{001111}_{{\bf B}_{(1,6)}} \underbrace{110000}_{{\bf B}_{(2,6)}} \underbrace{011001}_{{\bf B}_{(3,6)}}.
\end{equation*}
We verify that the weight in each subblock is within $[2,5]$, and hence, $\bx \in {\cal S}(18,6,[2,5])$. However, $\bx\notin {\cal W}(18,6,[2,5])$, since there are windows of size six, for example ${\bw_{(3,6)}}$ and ${\bw_{(9,6)}}$, that violate the weight constraint, 
\begin{equation*}
\bx= 00 {\color{red}{\underbrace{111111}_{\bw_{(3,6)}}}} {\color{blue}{\underbrace{000001}_{\bw_{(9,6)}}}} 1001.
\end{equation*}
\end{example}

\subsection{Our Contributions}
In this work, we design efficient methods of mapping arbitrary users' data to codewords in SECCs, SWCCs, bounded SECCs, and bounded SWCCs. Formally, for $0\le a<b\le \ell$, 
\begin{enumerate}[(i)]
\item In Section II, we propose an efficient encoder for bounded SECCs ${\cal S}(n,\ell,[a,b])$ using the Knuth's balancing technique. Note that ${\cal S}(n,\ell,a) \equiv {\cal S}(n,\ell,[a,\ell])$, and hence, the method can be applied to construct efficient encoder for SECCs ${\cal S}(n,\ell,a)$ as well. Particularly, we extend the Knuth's balancing technique for balanced codes, i.e. $a=b=\ell/2$, to construct ${\cal S}(n,\ell,[a,b])$ for a special case when $a=p_1\ell, b=p_2\ell, 0\le p_1<1/2< p_2\le 1$, and then generalize this technique for arbitrary $0\le a< b\le \ell$. 

\item In Section III, we first study the size of SWCCs ${\cal W}(n,\ell,a)$ and bounded SWCCs ${\cal W}(n,\ell,[a,b])$. When $a=p_1\ell, b=p_2\ell,$ for some $0\le p_1<1/2< p_2\le 1$, we show that when the window size satisfies certain constraints, the code size is at least $2^{n-1}$. We then propose efficient encoders for ${\cal W}(n,\ell,a)$ and SWCCs ${\cal W}(n,\ell,[a,b])$ by using the sequence replacement technique. For certain values of $a,b,\ell$, our method incurs only one redundant bit. 

\item In Section IV, we study these codes with given error correction capability. Particularly, we construct codes that can correct multiple errors with the assumption that the distance between any two errors is at least $\ell$. The intuition behind this assumption is that when the energy constraint is enforced over every window of size $\ell$, the probability of having error is minimized over every window. We consider the worst case scenario when there is at most one error over every window of size $\ell$.
\end{enumerate}

\section{SECCs and Bounded SECCs}
In this section, we propose simple coding scheme to construct ${\cal S}(n,\ell,a)$ and ${\cal S}(n,\ell,[a,b])$. We are interested in the case where the number of subblocks is constant, i.e. $m=\Theta(1), \ell=\Theta(n)$. Particularly, we first modify the Knuth's balancing technique to construct ${\cal S}(n,\ell,[a,b])$ when there exist two constant numbers $p_1,p_2$, $0\le p_1< 1/2< p_2\le 1$ such that $a\le p_1 \ell, b\ge p_2\ell$. We then extend this method to construct ${\cal S}(n,\ell,[a,b])$ and ${\cal S}(n,\ell,a)$ for arbitrary $a,b$.   

\subsection{Maximum Information Rate}
The following result is immediate.
\begin{proposition}
For $n=m\ell, 0\le a<b\le \ell$, we have $|{\cal S}(n,\ell,[a,b])|=\left( \sum_{i=a}^{b} {\ell \choose i} \right)^m$ and $|{\cal S}(n,\ell,a)|=\left( \sum_{i=a}^{\ell} {\ell \choose i} \right)^m$. 
\end{proposition} 

In fact, we are able to show that the sizes of ${\cal S}(n,\ell,a)$ and ${\cal S}(n,\ell,[a,b])$, under certain conditions of $a,b,\ell$, are at least $2^{n-1}$, and therefore, the channel capacity in such cases is 1. 

\begin{theorem}\label{theorem1} For all $0\le p_1 <1/2 <p_2\le 1$ and $a\le p_1\ell, b\ge p_2\ell$, let $c=\min\{1/2-p_1, p_2-1/2\}$. For $n,\ell$ that $(1/c^2) \log_e n \le \ell \le n$, we have $|{\cal S}(n,\ell,[a, b])| \geq 2^{n-1}$.
\end{theorem}

We defer the proof of Theorem~\ref{theorem1} to Section III, Theorem~\ref{bound}. In fact, Theorem~\ref{bound} presents a stronger result that under the assumption of $n,\ell,a,b,$ as mentioned in Theorem~\ref{theorem1}, we have $|{\cal W}(n,\ell,[a, b])| \geq 2^{n-1}$. Since ${\cal W}(n,\ell,[a, b]) \subset {\cal S}(n,\ell,[a, b])$, Theorem~\ref{theorem1} is then proved.


\subsection{Efficient Construction of ${\cal S}(n,\ell,[a,b])$}
In this section, we modify the Knuth's balancing technique to construct ${\cal S}(n,\ell,[a,b])$. Knuth's balancing technique is a linear-time algorithm that maps a binary message $\bx$ to a balanced binary word $\by$ of the same length by flipping the first $t$ bits of $\bx$ \cite{knuth}. The crucial observation demonstrated by Knuth is that such an index $t$ always exists and $t$ is commonly referred to as a {\em balancing index}. To represent such a balancing index, Knuth appends $\by$ with a short balanced suffix $\bp$ of length $\log n$ and hence, a lookup table of size $n$ is required. Modifications of the generic scheme are discussed in \cite{1988,imbalanced,bose1996,kas:2020}.

\begin{definition}
For a binary sequence $\bx \in \{0,1\}^n$ and $0\le t\le n$, let ${\bf f}_t(\bx)$ denote the binary sequence obtained by flipping the first $t$ bits of $\bx$.
\end{definition}

\begin{example}
Let $\bx=001111 \in \{0,1\}^{6}$. We have ${\bf f}_1(\bx)=101111, {\bf f}_2(\bx)=111111, {\bf f}_3(\bx)=110111, {\bf f}_4(\bx)=110011, {\bf f}_5(\bx)=110001,$ and ${\bf f}_6(\bx)=110000$. Hence, $t=5$ is the unique balancing index of $\bx$. In general, the balancing index may not be unique. For example, consider $\by=001100$. We observe that both ${\bf f}_1(\by)=101100$ and ${\bf f}_5(\by)=110010$ are balanced, therefore, both $t=1$ and $t=5$ are balancing indices of $\by$. 
\end{example}

We now extend the Knuth's method to construct ${\cal S}(n,\ell,[a,b])$ when $a\le p_1 \ell, b\ge p_2\ell$ for some constant $p_1,p_2$, $0\le p_1< 1/2< p_2\le 1$. Our main idea is that, for a codeword $\bx$ and the $i$th subblock ${\bf B}_{(i,\ell)}(\bx)$, there exists an index $t$ such that after flipping the first $t$ bits in ${\bf B}_{(i,\ell)}(\bx)$, the weight of the subblock is within $[p_1\ell,p_2\ell]$. We now show that such an index t always exists and there is an efficient method to find $t$. 

\begin{definition}\label{walk}
Let $n$ be even and set $[n]=\{0,1,2,\ldots,n\} $. For arbitrary $0<k\le n$, a {\em walk} of size $k$ in $[n]$ is the set of indices ${\rm S}_{(k, n)}\triangleq \{0,n\} \cup \{ik: i\ge 1 \mbox{ and } ik < n \}$.
\end{definition}


\begin{theorem}\label{epsilon-balanced}
Given $\ell$ even and $0\le p_1 <1/2 <p_2\le 1$. Let $k=(p_2-p_1)\ell$. For an arbitrary binary sequence $\bx \in \{0,1\}^{\ell}$, there exists an index $t$ in the set ${\rm S}_{(k,\ell)}$, such that the weight of ${\bf f}_t(\bx)$ is within $[p_1\ell,p_2\ell]$.
\end{theorem}

\begin{proof}
In the trivial case, when the weight of $\bx$ satisfies the constraint, i.e. ${\rm wt}(\bx) \in [p_1\ell,p_2\ell]$, then we can select $t=0 \in {\rm S}_{(k,\ell)}$. Otherwise, assume that ${\rm wt}(\bx) \notin [p_1\ell,p_2\ell]$, and without loss of generality, assume that $ {\rm wt}(\bx) < p_1\ell \le \ell/2$. Since ${\rm wt}(\bx) < \ell/2$, we have ${\rm wt}({\bf f}_{\ell}(\bx)) > \ell/2$. Now, for $k=(p_2-p_1)\ell$, consider the list of indices, $t_1=k, t_2=2k,$ and $t_i=ik \in {\rm S}_{(k, n)}$. Since  ${\bf f}_{t_i}(\bx)$ and  ${\bf f}_{t_{i+1}}(\bx)$ differ at most $k$ positions, and ${\rm wt}(\bx) < \ell/2$, ${\rm wt}({\bf f}_{\ell}(\bx)) > \ell/2$, there must be an index $t \in {\rm S}_{(k,\ell)}$ such that $ p_1\ell \le{\rm wt}({\bf f}_t(\bx)) \le p_2\ell$.
\end{proof}
\begin{example}
Let $\bx=110000000000 \in \{0,1\}^{12}, {\rm wt}(\bx)=2$. Let $p_1=1/3$ and $p_2=2/3$, i.e. we want a codeword that has weight in $[4,6]$. We compute $k=(p_2-p_1)\ell=4$. The set ${\rm S}_{(k,\ell)}=\{0,4,8,12\}$. We can verify that
\begin{align*}
{\bf f}_4(\bx)&=001100000000, \\
{\bf f}_8(\bx)&=001111110000, \\
{\bf f}_{12}(\bx)&=001111111111. 
\end{align*} 
Hence, for $t=8 \in {\rm S}_{(k,\ell)}$, we get ${\rm wt}({\bf f}_t(\bx))\in[p_1\ell,p_2\ell]$.
\end{example}
\begin{lemma}\label{extend}
Given $n>0$ and $0\le p_1 <1/2 <p_2\le 1$. Let $\bx\in\{0,1\}^n$ such that $ p_1 n \le {\rm wt}(\bx)\le p_2 n$. For any binary balance word $\by\in\{0,1\}^m$, we have $ p_1(n+m)\le{\rm wt}(\bx\by) \le p_2(n+m)$.
\end{lemma}
\begin{proof}
We have ${\rm wt}(\bx\by)={\rm wt}(\bx)+m/2$. Since $0\le p_1 <1/2 <p_2\le 1$, it implies that $p_1m<m/2<p_2m$.

 As given, $ p_1 n \le {\rm wt}(\bx)\le p_2 n$, we then conclude that $ p_1(n+m)\le{\rm wt}(\bx\by) \le p_2(n+m). $
\end{proof}

For constant $0\le p_1 <1/2 <p_2\le 1$, $k=(p_2-p_1)\ell$, the size of ${\rm S}_{(k,\ell)}$ is at most $\floor{1/(p_2-p_1)}+1$, which is independent of $\ell$. Let $r=\ceil{\log \left( \floor{1/(p_2-p_1)}+1 \right)}$. To encode an arbitrary binary sequence $\bx$ to a codeword in ${\cal S}(n,\ell,[a,b])$, where $a\le p_1 \ell, b\ge p_2\ell$, we divide $\bx$ into subblocks of length $N=\ell-r$. We then encode each subblock and concatenate the outputs. For each subblock, we simply find the smallest index $t$ in ${\rm S}_{(k,\ell-r)}$ such that $\by={\bf f}_t(\bx)$ satisfies the weight constraint. According to Theorem~\ref{epsilon-balanced}, such index always exists. To represent such index, we also append a short balanced suffix, and so, a lookup table of size $\log |{\rm S}_{(k,\ell-r)}|=r$ is required. 

For completeness, we describe the formal encoder/decoder of ${\cal S}(n,\ell,[a, b])$ as follows.
\vspace{0.05in}

\noindent{\bf Preparation phase.} Given $n=m\ell$, $0\le p_1 <1/2 <p_2\le 1$, set $k=(p_2-p_1)\ell$ and $r=\ceil{\log (\floor{1/(p_2-p_1)}+1)}$. Set ${\rm S}_{(k,\ell-r)}$ be the set of indices as defined in Definition~\ref{walk}. We construct a one-to-one correspondence between the indices in ${\rm S}_{(k,\ell-r)}$ and the $r$ bits balanced sequences.
\vspace{0.05in}

\noindent{\bf Encoder S}
\vspace{0.05in}

{\sc Input}: $\bx\in \{0,1\}^{m(\ell-r)}$\\
{\sc Output}: $\by = \enc_{\rm S}(\bx) \in {\cal S}(n,\ell,[a,b])$ where $a\le p_1\ell, b\ge p_2\ell$\\[-3mm]

\begin{enumerate}[(I)]
\item For $1\le i\le m$ Do:
\begin{itemize}
\item Set $\bz_i={\bf B}_{(i,\ell-r)}(\bx)$
\item Search for the first index $t$ in ${\rm S}_{(k,\ell-r)}$, such that ${\rm wt}({\bf f}_t(\bz_i)) \in [p_1(\ell-r),p_2(\ell-r)]$
\item Let $\bp_i$ be the $r$ bits balanced sequence representing $t$ 
\item Set $\by_i={\bf f}_t(\bz_i) \bp_i$
\end{itemize}

\item Finally, we output $\by=\by_1\by_2\ldots\by_m$
\end{enumerate}

\begin{theorem}
The Encoder S is correct. In other words, $\enc_{\rm S}(\bx) \in {\cal S}(n,\ell,[a,b])$ for all $\bx \in \{0,1\}^{m(\ell-r)}$.
\end{theorem}
\begin{proof}
To show $\by=\enc_{\rm S}(\bx) \in {\cal S}(n,\ell,[a,b])$, we need to verify that the weight of every subblock of $\by$ is in $[a,b]$. From Encoder S, the $i$th subblock is $\by_i={\bf f}_t(\bz_i) \bp_i$. Since ${\rm wt}({\bf f}_t(\bz_i))\in [p_1(\ell-r),p_2(\ell-r)]$ and $\bp_i$ is a balanced word of length $r$, according to Lemma~\ref{extend}, ${\rm wt}(\by_i) \in [p_1\ell, p_2\ell] \subseteq[a,b]$.
\end{proof}


\noindent{\bf Decoder S}
\vspace{0.05in}

{\sc Input}: $\by\in {\cal S}(n,\ell,[a, b])$ where $a\le p_1\ell, b\ge p_2\ell$\\
{\sc Output}: $\bx = \dec_{\rm S}(\by) \in \{0,1\}^{m(\ell-r)}$\\[-3mm]

\begin{enumerate}[(I)]
\item For $1\le i\le m$ Do:
\begin{itemize}
\item Set $\bz_i={\bf B}_{(i,\ell)}(\by)$
\item Let $\bp_i$ be the suffix of length $r$ of $\bz_i$ that corresponds to an index $t\in {\rm S}_{(k,\ell-r)}$  '
\item Obtain $\bz'_i$ by removing $\bp_i$ from $\bz_i$
\item Set $\bx_i={\bf f}_t(\bz'_i)$
\end{itemize}

\item Finally, we output $\bx=\bx_1\bx_2\ldots\bx_m$
\end{enumerate}
\vspace{0.05in}

 Alternatively, for each subblock, the index can be encoded/decoded in linear-time without the look-up table for ${\rm S}_{(k,\ell-r)}$. However, the redundancy increases from $r$ to $2r$ and the set of indices is ${\rm S}_{(k,\ell-2r)}$. Recall that $|{\rm S}_{(k,\ell-r)}|=|{\rm S}_{(k,\ell-2r)}|=r$. The modified Encoder S' can be constructed as follows. We skip the detail of the corresponding Decoder S'.
 \vspace{0.05in}

\noindent{\bf Encoder S'}. 
\vspace{0.05in}

{\sc Input}: $\bx\in \{0,1\}^{m(\ell-2r)}$\\
{\sc Output}: $\by = \enc_{\rm S}(\bx) \in {\cal S}(n,\ell,[a,b])$ where $a\le p_1\ell, b\ge p_2\ell$\\[-3mm]

\begin{enumerate}[(I)]
\item For $1\le i\le m$ Do:
\begin{itemize}
\item Set $\bz_i={\bf B}_{(i,\ell-2r)}(\bx)$
\item Search for the first index $t$ in ${\rm S}_{(k,\ell-2r)}$, such that ${\rm wt}({\bf f}_t(\bz_i)) \in [p_1(\ell-2r),p_2(\ell-2r)]$
\item Let $\Gamma=\tau_1\tau_2\ldots\tau_r$ be the binary representation of the rank of the index $t$ in ${\rm S}_{(k,\ell-2r)}$
\item Set $\bp_i=\Gamma\overline{\Gamma}$ of length $2r$, where $\overline{\Gamma}$ is the complement of $\Gamma$ and set $\by_i={\bf f}_t(\bz_i) \bp_i$
\end{itemize}
\item Finally, we output $\by=\by_1\by_2\ldots\by_m$
\end{enumerate}
\vspace{0.05in}

\noindent{\bf Analysis.} The redundancy for encoding each subblock in Encoder S (or Encoder S') is $r=\ceil{\log \left( \floor{1/(p_2-p_1)}+1 \right)}$ (or $2r$), which is independent of $\ell$. In other words, for constant $p_1,p_2$, $r=\Theta(1)$. Consequently, the total redundancy for codewords of length $n=m\ell$ is then $mr=\Theta(m)$. Therefore, this encoding method is efficient for large $\ell$ and the number of subblocks is small, compared to the length of codeword, i.e. $m=\Theta(1),\ell=\Theta(n)$. In such cases, the rate of Encoder S is $(n-mr)/n=1- mr/n \to 1$, and similarly, the rate of Encoder S' is $ 1- 2mr/n \to 1$, both approaching the channel capacity. Indeed, the same argument applies when $m=o(n)$. It is easy to verify that the complexity of Encoder/Decoder S (or S') are linear in the codeword length.  
\subsection{Extension to ${\cal S}(n,\ell,a)$}
We can modify Encoder S to construct ${\cal S}(n,\ell,a)$ or ${\cal S}(n,\ell,[a, b])$ for arbitrary $0\le a<\ell/2<b\le \ell$.

For ${\cal S}(n,\ell,[a,b])$, we let $k=b-a$ and the set $S_{(k,\ell)}$ is of size at most $\floor{\ell/(b-a)}+1$. The redundancy to encode each subblock of size $\ell$ is then $\ceil{\log (\floor{\ell/(b-a)}+1)}$. The efficiency of the encoder is high when $a=o(\ell)$ or $b=\Theta(\ell)$. In such cases, since $b-a=\Theta(\ell)$, we have $\ceil{\log (\floor{\ell/(b-a)}+1)}=\Theta(1)$. 

Particularly, for SECCs ${\cal S}(n,\ell,a)$ when $a< \ell/2$, Encoder S incurs only one redundant bit for each subblock. The simple idea is as follows. If the $i$th subblock has weight $w<a<\ell/2$, the encoder simply flips the whole subblock (or equivalently take its complement), the weight of the complement is then $w'> \ell/2> a$. The encoder appends one bit $\bp=1$ (or $0$) if the flipping action is needed (or not needed). This classic code is known as the {\em polarity bit code} \cite{book}. 

\begin{example}
Let $n=21, \ell=7, m=3$ and $a=3$. Suppose the source data is $\bx= 110000 011001 111100 \in \{0,1\}^{18}$. The encoder checks every subblock of length 6 and outputs $\bc \in \{0,1\}^{21}$ where,
\begin{equation*}
\bc=  \underbrace{001111{\color{red}{1}}}_{{\bf B}_{(1,7)}} \underbrace{011001{\color{blue}{0}}}_{{\bf B}_{(2,7)}} \underbrace{111100{\color{blue}{0}}}_{{\bf B}_{(3,7)}}.
\end{equation*}
To decode $\bc$, the decoder also checks every subblock of length 7, if the last bit is 1, it flips the prefix. We then obtain the source data $\bx$,
\begin{equation*}
\bx=  \underbrace{{\color{red}{110000}}}_{{\bf B}_{(1,6)}} \underbrace{{\color{blue}{011001}}}_{{\bf B}_{(2,6)}} \underbrace{{\color{blue}{111100}}}_{{\bf B}_{(3,6)}}.
\end{equation*} 
\end{example}
\vspace{0.05in}

\section{SWCCs and Bounded SWCCs}
In this section, we propose a simple coding scheme to construct ${\cal W}(n,\ell,[a,b])$ and ${\cal W}(n,\ell,a)$ by using the sequence replacement technique. Particularly, to construct ${\cal W}(n,\ell,[a,b])$ when $a\le p_1\ell, b\ge p_2\ell$ for some constant $0\le p_1 <1/2 <p_2\le 1$, our method incurs only one redundant bit. Since ${\cal W}(n,\ell,[a,b]) \subset {\cal S}(n,\ell,[a,b])$ for $n=m\ell$, this method also provides an efficient encoder for ${\cal S}(n,\ell,[a,b])$ with only one redundant bit. This yields a significant improvement in coding redundancy with respect to the Knuth's balancing technique described in Section II. Note that the efficiency of Encoder S is high when the number of subblocks is constant, i.e. $m=\Theta(1), \ell=\Theta(n)$, since the redundancy grows linearly with $m$. In this section, we show that there exists an efficient encoder when $m$ is a function of $n$, as long as the size of subblocks (or windows) satisfies certain constraints. 
\subsection{Maximum Information Rate}
The following result  implies that there exist such codes with size at least $2^{n-1}$ and hence, approaching the channel capacity. 
 
\begin{theorem}\label{bound} For $0\le p_1 <1/2 <p_2\le 1$, let $c=\min\{1/2-p_1, p_2-1/2\}$. For $a\le p_1\ell, b\ge p_2\ell,$ and $(1/c^2) \log_e n \le \ell \le n$, we have $|{\cal W}(n,\ell,[a,b])| \geq 2^{n-1}$.  
\end{theorem}

To prove Theorem~\ref{bound}, we require {\em Hoeffding's inequality} \cite{Hoeffding.1963}.

\begin{theorem}[Hoeffding's Inequality]
Let $Z_1, Z_2,\ldots, Z_n$ be independent bounded random variables such that $a_i\le Z_i\le b_i$ for all $i$.
Let $S_n=\sum_{i=1}^n Z_i$. For any $t>0$, we have
\begin{equation*}
P(S_n-E[S_n]\geq t) \leq e^{-{2t^2}/{\sum_{i=1}^n (b_i-a_i)^2}}.
\end{equation*}
\end{theorem}

\begin{proof}[Proof of Theorem~\ref{bound}]  
Let $\bx$ be uniformly at random selected element from $\{0,1\}^n$. A window ${\bw}_{(i,i+\ell-1)}$ of length $\ell$ of $\bx$ is said to be a {\em forbidden window} if the weight of it does not satisfy the constraint, i.e. ${\rm wt}({\bw}_{(i,i+\ell-1)}) \notin [a,b]$.
We evaluate the probability that the first window $\bw_{(0,\ell)}(\bx)$ 
is a forbidden window. Note that $[p_1\ell,p_2\ell] \subseteq [a,b]$. Applying Hoeffding's inequality we obtain:
\begin{align*}
P \left(  {\rm wt}(\bw_{(0,\ell)}) \notin [a,b]  \right) &\le P \left(  {\rm wt}(\bw_{(0,\ell)}) \notin [p_1\ell,p_2\ell]  \right) \\
&\le P \left( \left| {\rm wt}(\bw_{(0,\ell)}) - \ell/2\right| \geq c\ell\right) \\ 
&=2 P \left( {\rm wt}(\bw_{(0,\ell)}) - \ell/2 \geq c\ell\right) \\
&\leq 2 e^{-\frac{2c^2\ell^2}{\ell}}=2e^{-2c^2\ell}.
\end{align*}
The function $f(\ell)=2e^{-2c^2\ell}$ is decreasing in $\ell$. Since there are $(n-\ell+1) \le n$ windows, applying the union bound, we get
\begin{equation*}
P(\bx \notin {\cal W}(n,\ell,[a,b])) \leq n 2e^{-2c^2\ell} \le 2n e^{-2 \log_e n} = 2/n.
\end{equation*}
Therefore,
\begin{equation*}
|{\cal W}(n,\ell,[a,b])| \ge 2^n(1-2/n). 
\end{equation*}
For $n\geq 4$, we have that $1- 2/n \geq 1/2$. Therefore, $ {\cal W}(n,\ell,[a,b]) \geq 2^{n-1}$. Note that, since $\ell\le n$, we also require $n$ to be large enough such that $n\ge (1/c^2) \log_e n$. 
\end{proof}


Before we present the efficient encoder/decoder for ${\cal W}(n,\ell,[a,b])$, the following corollary is crucial to show the correctness of our algorithms. When $m=1$, by replacing $\ell$ with $(\ell-2)$ in Theorem~\ref{bound}, we obtain the following result. 
\begin{corollary}\label{coro} For $\ell\ge6$, and $\ell-2\ge (1/c^2) \log_e (\ell-2)$, we have $\big| {\cal W} \left( \ell-2,\ell-2,[p_1 (\ell-2), p_2 (\ell-2)] \right) \big| \geq 2^{\ell-3}$.
\end{corollary}

\subsection{Sequence Replacement Technique}
We first present an efficient encoder for ${\cal W}(n,\ell,[a,b])$ when there exist constant numbers $p_1,p_2$, $0\le p_1 <1/2 <p_2\le 1$ that $a\le p_1\ell, b\ge p_2\ell$. For simplicity, we construct an efficient map that translates arbitrary messages into codewords in ${\cal W}(n,\ell,[p_1\ell,p_2\ell]) \subseteq {\cal W}(n,\ell,[a,b])$. A similar class of SWCCs has been introduced in \cite{tt2019,gabrys2018}. Formally, such codes impose the weight constraint over every window of size at least $\ell$, and here we refer such codes as {\em strictly constrained SWCCs}. Some lower bounds on the size of codes are provided for specific value of $p_1$ and $p_2$ (for example, \cite{gabrys2018} considered $p_1=1/6$ and $p_2=5/6$).

Our method is based on the {\em sequence replacement technique}. The sequence replacement technique has been widely used in the literature \cite{immink2018, W2010, schoeny2017, O2019}. It is an efficient method for removing forbidden windows from a source word. In general, the encoder removes the forbidden windows and subsequently inserts its representation (which also includes the position of the windows) at predefined positions in the sequence. Crucial to the replacement step is to estimate the total number of forbidden windows. 

In the following of the section, for a binary sequence $\bx$, a window of size $\ell$ of $\bx$ is said to be an  {\em $\ell$-forbidden window} if the weight of this window does not belong to $[p_1\ell,p_2\ell]$. Let ${\bf F}(\ell,[p_1\ell,p_2\ell])$ denote the set of all $\ell$-forbidden windows of size $\ell$. The following theorem provides an upper bound on the size of ${\bf F}(\ell,[p_1\ell,p_2\ell])$.

\begin{theorem}\label{crucial1} For $0\le p_1 <1/2 <p_2\le 1$, let $c=\min\{1/2-p_1, p_2-1/2\}$. For $n\ge16$ and $\ell\le n$ such that $(1/c^2) \log_e n \le \ell$, let 
$k=\ell-3-\log n$, there exists an one-to-one map $\Phi: {\bf F}(\ell,[p_1\ell,p_2\ell]) \to \{0,1\}^k$.
\end{theorem} 

\begin{proof}
We first show that $k>0$. Since $c\le 1/2$, we have $\ell\ge (1/c^2) \log_e n \ge 4 \log_e n > 2.77 \log n >3+\log n$ for $n\ge 4$. 
For an arbitrary $\bx \in \{0,1\}^{\ell}$, from the proof of Theorem~\ref{bound}, we have
\begin{equation*}
P(\bx \in {\bf F}(\ell,[p_1\ell,p_2\ell])) \leq 2e^{-2c^2\ell} \le 2/n^2.
\end{equation*}
Therefore, the size of $ {\bf F}(\ell,[p_1\ell,p_2\ell])$ is at most
\begin{equation*}
|{\bf F}(\ell,[p_1\ell,p_2\ell])| \le (2/n^2) 2^\ell = 2^{\ell+1}/n^2. 
\end{equation*}
Thus, to represent all forbidden windows in $ {\bf F}(\ell,[p_1\ell,p_2\ell])$, we need all binary sequences of length at most $k'=\log {2^{\ell+1}/n^2} = \ell+1-2 \log n \le \ell-3-\log n=k$ for all $n\ge 16$. Therefore, there exists a one-to-one map $\Phi: {\bf F}(\ell,[p_1\ell,p_2\ell]) \to \{0,1\}^k$.
\end{proof}

The key idea in the sequence replacement technique is to ensure that the replacement procedure is guaranteed to terminate. The general idea is to replace each forbidden window of length $\ell$ (if there is) with a subsequence of length shorter than $\ell$. Consequently, after each replacement step, the length of codeword is reduced, the replacement procedure is guaranteed to terminate. In our problem, in the worst case, the final replacement step occurs when the length of the current word is $\ell+1$, since after another replacement (if needed), the length of the current word becomes at most $\ell$ and we cannot proceed further. This final step is crucial to ensure that the final output codeword satisfies the weight constraint. The following result provides an upper bound for the number of sequences of length $\ell+1$ that include at least a forbidden window. 

Let ${\bf G}(\ell+1,[p_1\ell,p_2\ell])$ denote the set of all binary sequences of length $(\ell+1)$ that contain at least one forbidden window.
\begin{theorem}\label{crucial2} 
For $0\le p_1 <1/2 <p_2\le 1$, let $c=\min\{1/2-p_1, p_2-1/2\}$. For $\ell\ge 7$ and $\ell\ge (1/c^2) \log_e (\ell+1)$, we have $|{\bf G}(\ell+1,[p_1\ell,p_2\ell])|\le 2^{\ell-3}$. In addition, there exists an one-to-one map $\Psi: {\bf G}(\ell+1,[p_1\ell,p_2\ell]) \to {\cal W}(\ell-2,\ell-2,[p_1 (\ell-2), p_2 (\ell-2)])$.
\end{theorem} 

\begin{proof}
Since there are only two windows, similar to the proof of Theorem~\ref{bound} and Theorem~\ref{crucial1}, by using the union bound, and for an arbitrary $\bx \in \{0,1\}^{\ell+1}$, we have
\begin{equation*}
P(\bx \in {\bf G}(\ell+1,[p_1\ell,p_2\ell])) \leq 2\times 2\times e^{-2c^2\ell} \le 4/(\ell+1)^2.
\end{equation*}
Therefore, the size of $ {\bf G}(\ell+1,[p_1\ell,p_2\ell])$ is at most
\begin{equation*}
|{\bf G}(\ell+1,[p_1\ell,p_2\ell])| \le (4/(\ell+1)^2) 2^{\ell+1} = 2^{\ell+3}/(\ell+1)^2. 
\end{equation*}
For all $\ell\ge7$, we have $2^{\ell+3}/(\ell+1)^2\le 2^{\ell-3}$.
According to Corollary~\ref{coro}, $|{\cal W}(\ell-2,\ell-2,[p_1 (\ell-2), p_2 (\ell-2)])| \ge 2^{\ell-3}$, which implies that there exists an one-to-one map $\Psi: {\bf G}(\ell+1,[p_1\ell,p_2\ell]) \to {\cal W}(\ell-2,\ell-2,[p_1 (\ell-2), p_2 (\ell-2)])$.
\end{proof}

\subsection{Efficient Encoder/Decoder for ${\cal W}(n,\ell,[a,b])$}
We now present a linear-time algorithm to encode ${\cal W}(n,\ell,[a,b])$. For simplicity, we assume $\log n$ is an integer. 
\vspace{0.05in}

\noindent{\bf Encoding algorithm.}
The algorithm contains three phases: {\em initial phase}, {\em replacement phase} and {\em extension phase}. Particularly, the replacement phase includes {\em regular replacement} and {\em special replacement}.
\vspace{0.05in}

\noindent{\bf Initial phase.} The source sequence $\bx\in\{0,1\}^{n-1}$ is prepended with $0$, to obtain $\by=0\bx \in \{0,1\}^n$. The encoder scans $\by$ and if there is no forbidden window, it outputs $\by$. Otherwise, it proceeds to the replacement phase.
\vspace{0.05in}

\noindent{\bf Replacement phase.} The aim of this procedure is that, at the end of the replacement phase, all forbidden windows of size $\ell$ will be removed and the length of the current word is at least $\ell$. If the length of the current word is larger than $\ell+1$, the encoder proceeds to the regular replacement. On the other hand, if the length of the current word $\by$ is $(\ell+1)$, the encoder proceeds to the special replacement. 

\begin{itemize}
\item {\bf Regular replacement.} Let $\bw_{(i,\ell)}$ be the first forbidden window in $\by$, for some $1\le i\le n-\ell+1<n$. According to Theorem~\ref{crucial1}, the total number of forbidden windows of size $\ell$ is at most $2^{k}$, where $k=\ell-3-\log n$. Let $\bq_1$ be the binary representation of length $\log n$ of $i$, and $\bq_1=\Phi(\bw_{(i,\ell)})$ of length $k$. The encoder sets $\bq_{\rm regular}=\bq_1\bq_2$ and removes this forbidden window $\bw_{(i,\ell)}$ from $\by$, and then prepends $11\bq_{\rm regular}$ to $\by$. If, after this replacement, $\by$ contains no forbidden window, the encoder proceeds to the extension phase. Otherwise, the encoder repeats the replacement phase. Note that such an operation reduces the length of the sequence by one, since we remove $\ell$ bits and replace by $2+\log n+k=2+\log n+(\ell-3-\log n)=\ell-1$ (bits). Therefore, this procedure is guaranteed to terminate. 

\item {\bf Special replacement.} According to Theorem~\ref{crucial2}, the number of such words is at most $2^{\ell-3}$. The encoder sets  $\bq_{\rm special}=\Psi(\by)\in {\cal W}(\ell-2,\ell-2,[p_1 (\ell-2), p_2 (\ell-2)])$, i.e. ${\rm wt}(\bq_{\rm special})\in[p_1(\ell-2),p_2(\ell-2)]$, and then replaces all $(\ell+1)$ bits with $10\bq_{\rm special}$. After this replacement, the current word is of length $\ell$ and it does not contain any forbidden window. This is because the prefix is 10, which is balanced, the suffix $\bp_{\rm special}$ satisfies ${\rm wt}(\bq_{\rm special})\in[p_1(\ell-2),p_2(\ell-2)]$, therefore, according to Lemma 1, ${\rm wt}(10\bq_{\rm special}) \in [p_1\ell,p_2\ell]$. The encoder then proceeds to the extension phase.
\end{itemize} 


\noindent{\bf Extension phase.} If the length of the current sequence $\by$ is $n_0$ where $n_0<n$, the encoder appends a suffix of length $n_1=n-n_0$ to obtain a sequence of length $n$. Note that at the end of the replacement phase, the length of the current word is at least $\ell$. Let $\bz$ be the last window of size $\ell$ in $\by$. Suppose that $\bz=z_1z_2\ldots z_\ell$ and ${\rm wt}(\bz)\in [p_1\ell,p_2\ell]$. A simple way to create a suffix is to repeat appending $\bz$ for sufficient times until the length exceeds $n$. Let $j$ be the smallest integer such that $\bc=\by \bz^{j}$ is of length greater than $n$. The encoder outputs the prefix of length $n$ of $\bc$. We now show that $\bc \in {\cal W}(n,\ell,[p_1\ell,p_2\ell])$. Since $\by$ does not contain any forbidden window, it remains to show that there is no forbidden windows in the suffix $\bz^{j}$. It is easy to see that repeating the vector $\bz$ clearly satisfies the constraint since every window of size $\ell$ generated in this manner is a cyclic shift of the vector $\bz$, and since ${\rm wt}(\bz) \in [p_1\ell,p_2\ell]$, there is no forbidden window. 
\vspace{0.05in}

We now present an efficient algorithm to decode the source data uniquely. The decoding procedure is relatively simple as follows. 

\noindent{\bf Decoding algorithm.}
The decoder scans from left to right. If the first bit is 0, the decoder simply removes 0 and identifies the last $(n-1)$ bits are source data. On the other hand, if it starts with 11, the decoder takes the prefix of length $(\ell-1)$ and concludes that this prefix is obtained by a regular replacement. In other words, the prefix is of the form $11\bq_{\rm regular}$, $\bq_{\rm regular}=\bq_1\bq_2$ where $\bq_1$ is of length $\log n$ and $\bq_2$ is of length $k$. The decoder removes this prefix, adds the forbidden window $\bw=\Phi^{-1}(\bq_2)$ into position $i$, which takes $\bq_1$ as the binary representation. However, if it starts with 10, the decoder takes the prefix of length $\ell$ and concludes that this prefix is obtained by a special replacement. In other words, the prefix of length $\ell$ can be represented by $10\bq_{\rm special}$. The decoder replaces the prefix of length $\ell$ with the window of length $\ell+1$, $\bw=\Psi^{-1}(\bq_{\rm special})$, and then proceeds to decode from $\bw$. It terminates when the first bit is 0, and the decoder simply takes the following $(n-1)$ bits as the source data.
\vspace{0.05in}

We summary the details of our proposed encoder/decoder for ${\cal W}(n,\ell,[p_1\ell,p_2\ell])$ as follows. 
\vspace{0.05in}

\noindent{\bf Preparation.} Given $0\le p_1 <1/2 <p_2\le 1$, $c=\min\{1/2-p_1, p_2-1/2\}, n\ge16, \ell\ge7, n\ge \ell$, where $\ell-2\ge (1/c^2) \log_e \ell$. Let $k=\ell-3-\log n$, we construct two one-to-one maps: 
\begin{small}
\begin{align*}
&\Phi: {\bf F}(\ell,[p_1\ell,p_2\ell]) \to \{0,1\}^k, \text{ and } \\
&\Psi: {\bf G}(\ell+1,[p_1\ell,p_2\ell]) \to {\cal W}(\ell-2,\ell-2,[p_1 (\ell-2), p_2 (\ell-2)]).
\end{align*}
\end{small}
In other words, every forbidden window of size $\ell$ in ${\bf F}(\ell,[p_1\ell,p_2\ell])$ is represented by a $k$ bits sequence, and every window of size $\ell+1$ in ${\bf G}(\ell+1,[p_1\ell,p_2\ell])$ is represented by a $\ell$ bits sequence in ${\cal W}(\ell-2,\ell-2,[p_1 (\ell-2), p_2 (\ell-2)])$. 

\noindent{\bf Encoder W}
\vspace{0.05in}

{\sc Input}: $\bx\in \{0,1\}^{n-1}$\\
{\sc Output}: $\bc = \enc_W(\bx) \in {\cal W}(n,\ell,[p_1\ell,p_2\ell]) $\\[-2mm]
\begin{enumerate}[(I)]
\item  {\em Initial Phase.} Set $\by \leftarrow 0 \bx$
\item {\em Replacement Phase.}
 
{\bf While} (there is forbidden window in $\by$) and (the length of $\by$ is greater than $\ell+1$) {\bf Do}:
\begin{itemize}
\item Let $i$ be smallest index such that $\bw_{(i,\ell)}$ is forbidden, $1\le i\le n-\ell+1$ 
\item Let $\bq_1$ be the binary representation of length $\log n$ of $i$ and let $\bq_2=\Phi(\bw_{(i,\ell)}) \in \{0,1\}^k$ 
\item Set $\bq_{\rm regular}=\bq_1 \bq_2$ 
\item Set $\by \leftarrow \by$ removes $\bw_{(i,\ell)}$
\item Set $\by \leftarrow 1 1 \bq_{\rm regular} \by $
\end{itemize}

{\bf If} (the length of $\by$ is $(\ell+1)$) and (there is a forbidden window in $\by$) {\bf then}:
\begin{itemize}
\item Set $\bq_{\rm special}= \Psi(\by)$
\item Set $\by \leftarrow 1 0 \bq_{\rm special}$
\end{itemize}
\item {\em Extension Phase}. 
\begin{itemize}
\item Set $\bz$ be the last window of size $\ell$ in $\by$
\item Let $j$ be the smallest integer where $\bc=\by \bz^{j}$ is of length greater than $n$
\end{itemize}
\item Output the prefix of length $n$ of $\bc$
\end{enumerate}


\noindent{\bf Decoder W}.
\vspace{0.05in}

{\sc Input}: $\bc\in {\cal W}(n,\ell,[p_1\ell,p_2\ell]) $\\
{\sc Output}: $\bx = \dec_W(\bc) \in \{0,1\}^{n-1} $\\[-2mm]
\begin{enumerate}[(I)]

\item  
{\bf While} (the first bit is not 0) {\bf Do}:
\begin{itemize}
\item {\bf If} (the first two bits are 11) {\bf then}:  
\begin{enumerate}[(i)]
\item Let $11\bq_1\bq_2$ be the prefix of length $\ell-1$ of $\bc$ where $\bq_1$ is of length $\log n$ and $\bq_2$ is of length $k$
\item $\bc \leftarrow \bc$ remove the prefix 
\item Let $i$ be the index whose binary representation is $\bq_1$ 
\item Let $\bw$ be the forbidden window of size $\ell$ in ${\bf F}(\ell,[p_1\ell,p_2\ell])$, $\bw=\Phi^{-1}(\bq_2)$
\item Update $\bc$ by adding $\bw$ into $\bc$ at index $i$
\end{enumerate}
\item {\bf If} (the first two bits are 10) {\bf then}:
\begin{enumerate}[(i)]
\item Let $10\bq_{\rm special}$ be the prefix of length $\ell$ of $\bc$
\item Let $\bw \in {\bf G}(\ell+1,[p_1\ell,p_2\ell])$ such that $\bw=\Psi^{-1}(\bq_{\rm special})$
\item Set $\bc \leftarrow \bw$ 
\end{enumerate}
\end{itemize}

\item {\bf If} (the first bit is 0) {\bf then}: 
\begin{itemize}
\item Remove 0
\item Set $\bc$ be the prefix of length $n-1$
\end{itemize}
\item Output $\bc$
\end{enumerate}

\noindent{\bf Complexity Analysis.} For codewords of length $n$, it is easy to verify that Encoder W and Decoder W have linear-time complexity. Particularly, in Encoder W, the initial phase takes $\Theta(1)$ time. The total number of replacement in the replacement phase is $\Theta(n)$, and hence, the running time of replacement phase is $\Theta(n)$. The extension phase takes $\Theta(n)$ time. Therefore, the running time of Encoder W is $\Theta(n)$. Decoder W does the reverse procedure of Encoder W, and therefore, the running time is also $\Theta(n)$. Even though Encoder W offers lower redundancy than Encoder S, it suffers from more severe error propagation, i.e. during the decoding procedure, a small number of corrupted bits at the channel output might result in error propagation that could corrupt a large number of the decoded bits. On the other hand, Encoder S (Decoder S) encodes (decodes) subblocks separately and concatenates the outputs, and hence has a limited error propagation.

\subsection{Extension to ${\cal W}(n,\ell,a)$, ${\cal S}(n,\ell,[a,b])$, and ${\cal S}(n,\ell,a)$}
Encoder W can be used to construct SWCCs ${\cal W}(n,\ell,a)$ for $a<L/2$ with high efficiency. Especially, when $a\ll L$, i.e. there exist a constant $p_1$ that $a<p_1L$ for some $p_1< 1/2$, we can set $p_2=1$ and use Encoder W to construct ${\cal W}(n,\ell,[a,\ell])$.

Since ${\cal W}(n,\ell,[a,b]) \subset {\cal S}(n,\ell,[a,b])$ for $n=m\ell$, this method also provides an efficient encoder for ${\cal S}(n,\ell,[a,b])$ with only one redundant bit. This yields a significant improvement in coding redundancy with respect to the Knuth's balancing technique described in Section II. Recall that, for codewords of length $n=m\ell$, the redundancy of Encoder S is $\Theta(m)$. In contrast, the redundancy of Encoder W remains one bit for large value of $m$ as long as $\ell$ is sufficient large (refer to the preparation step in Encoder W). For example, one may set $\ell=\Theta(\log n)$ and $m=\Theta(n/\log n)$, Encoder W incurs only 1 redundant bit.

In addition, recall that ${\cal S}(n,\ell,a) \equiv {\cal S}(n,\ell,[a,b])$, therefore, for $a<\ell/2$, Encoder W can also be used to construct SECCs ${\cal S}(n,L,a)$ by setting $p_2=1$. Similarly, Encoder W can be easily modified to handle the case $a=\ell/2$.

\section{Error-Correction Codes}

In this section, we combine the previous constructions of Encoder S and Encoder W with error correction constraints. The output codewords satisfy the weight constraint and are capable of correcting multiple substitution errors. In this work, we assume that the distance between two errors is at least $\ell$. The intuition behind this assumption is that, since the energy constraint (or weight constraint) is guaranteed over every subblocks (or windows) of length $\ell$, the probability of having multiple errors in a subblock or window is small. A similar model correcting single deletion or single insertion over subblocks has been studied by Abroshan \et{} \cite{M:2017}. In this work, we impose the Hamming distance constraint and the codebooks are capable of correcting substitution errors. We first introduce the {\em Varshamov-Tenengolts (VT) codes} defined by Levenshtein \cite{le1965} to correct a single substitution. 

\begin{definition} The {\em binary VT syndrome} of a binary sequence $\bx\in\{0,1\}^n$ is defined to be
${\rm Syn}(\bx)=\sum_{i=1}^n i x_i$.
\end{definition}

For $a \in  \bbZ_{n}$, let
\begin{equation*}\label{VTsub}
{\rm L}_{a}(n)=\left\{\bx\in \{0,1\}^n: {\rm Syn}(\bx) = a \ppmod{2n}\right\}.
\end{equation*}

\begin{theorem}[Levenshtein \cite{le1965}]\label{thm:lev}
For $a \in  \bbZ_{2n}$, the code ${\rm L}_{a}(n)$ can correct a single substitution in linear time. There exists a linear-time decoding algorithm $\dec^{\rm L}_a:\{0,1\}^{n}\to {\rm L}_a(n)$ such that the following holds.
If $\bc\in {\rm L}_a(n)$ and $\by$ is the received vector with at most one substitution,
then $\dec^{\rm L}_a(\by)=\bc$.
\end{theorem}

In fact, Levenshtein \cite{le1965}  showed that ${\rm L}_{a}(n)$ can also correct a single deletion or single insertion. 

\subsection{Construction of SECCs with Error-Correction Capability}

In SECC ${\cal S}(n,\ell,[a,b])$ or ${\cal S}(n,\ell,a)$, each codeword contains $m=n/\ell$ subblocks of length $\ell$. We simply append the information of the syndrome of each subblock to the end of each subblock. Note that the redundant part must also satisfy the weight constraint. To do so, we propose a simple method to ensure the redundant part is balanced. The extra redundancy for each subblock is $2\log 2\ell$, and hence, the total redundancy of the encoder is $2m \log 2\ell$. For simplicity, assume that $t=\log 2\ell$ is integer. In the following, we present an efficient encoder for ${\cal S}(n,\ell,[a,b])$ that can correct $m$ substitution errors. 
For simplicity, we first present the case where $a\le p_1\ell, b\ge p_2\ell$ for some constant $0\le p_1<1/2<p_2\le \ell$. This construction can be easily modified to handle other classes of SECCs (for arbitrary parameters $a,b$ or ${\cal S}(n,\ell,a)$ where $a\le \ell/2$, refer to Subsection II-C).
\vspace{0.05in}

\noindent{\bf Preparation phase.} Given $n=m\ell$, $\ell_1=\ell-2\log 2\ell-r$, $k=(p_2-p_1)\ell_1, r=\ceil{\log (\floor{1/(p_2-p_1)}+1)}$, set ${\rm S}_{(k,\ell_1)}$ be the set of indices as defined in Definition~\ref{walk}. We construct a one-to-one correspondence between the indices in ${\rm S}_{(k,\ell_1)}$ and the $r$ bits balanced sequences. We require $\ell$ to be large enough so that $\ell_1=\ell-2\log 2\ell-r>0$. Note that $r=O(\log \ell)$.
\vspace{0.05in}

\noindent{${\textbf{Encoder S}}^{\rm ECC}$}. 

{\sc Input}: $\bx \in \{0,1\}^{m\ell_1}$\\
{\sc Output}: $\bc \triangleq \enc_{\cal S}^{\rm ECC}(\bx) \in {\cal S}(n,\ell,[a,b])$, where $a\le p_1\ell, b\ge p_2\ell$ and $n=m\ell$ \\[-2mm]

\begin{enumerate}[(I)]
\item Set $\ell_2=\ell-2\log 2\ell$. Use the Encoder S to obtain $\by = \enc_{\cal S}(\bx) \in {\cal S}(m\ell_2,\ell_2,[p_1\ell_2,p_2\ell_2])$. In other words, each subblock of length $\ell_1=\ell-2\log 2\ell-r$ in $\bx$ is encoded to a subblock of length $\ell_2=\ell-2\log 2\ell$ in $\by$ 

\item For $1\le i\le m$ Do:
\begin{itemize}
\item Set $\bz_i={\bf B}_{(i,\ell_2)}(\by)$
\item Compute $a={\rm Syn}(\bz_i) \ppmod{2\ell}$
\item Set $\bp$ be the binary representation of $a$ of length $\log 2\ell$
\item Set $\bq$ be the complement of $\bp$, i.e. $\bq=\overline{\bp}$
\item Set $\bc_i=\bz_i \bp \bq$ of length $\ell$
\end{itemize}

\item Output $\bc=\bc_1\bc_2\ldots\bc_m$
\end{enumerate}

\begin{theorem}
The ${\text{Encoder S}}^{{\rm ECC}}$ is correct. In other words, $\enc_{\cal S}^{\rm ECC}(\bx) \in {\cal S}(n,\ell,[a,b])$ and is capable of correcting at most $m$ substitution errors for all $\bx$ with the assumption that the distance between any two errors is at least $\ell$.
\end{theorem}
\begin{proof}
Let $\bc=\enc_{\cal S}^{\rm ECC}(\bx)$. We first show that $\bc \in {\cal S}(n,\ell,[a,b])$. Since $\bz_i={\bf B}_{(i,\ell_2)}(\by)$ where $\by \in {\cal S}(m\ell_2,\ell_2,[p_1\ell_2,p_2\ell_2])$, ${\rm wt}(\bz_i) \in [p_1\ell_2,p_2\ell_2]$. On the other hand, $\bp\bq$ is balanced since $\bq$ is the complement of $\bp$. According to Lemma~\ref{extend}, the $i$th subblock $\bc_i=\bz_i\bp\bq$ satisfy the weight constraint, i.e. ${\rm wt}(\bc_i) \in [p_1\ell,p_2\ell] \subseteq [a,b]$.

It remains to show that each subblock of $\bc$ can correct a substitution error. To do so, we provide an efficient decoding algorithm. Suppose that we receive a sequence $\by=\by_1\by_2\ldots\by_m$ where each subblock $\by_i$ is of length $\ell$. For $1\le i\le m$, we decode the $i$th subblock as follows. Let $\bz_i$ be the suffix of length $\ell_2=\ell-2\log 2\ell$ of $\by_i$, $\bp$ be the following $\log 2\ell$ bits, and $\bq$ be the suffix of length $\log 2\ell$.
\begin{itemize}
\item If $\bq\neq \overline{\bp}$, then we conclude that there is an error in the suffix $\bp\bq$, consequently there is no error in $\bz_i$. The decoder use Decoder S to decode $\bz_i$.
\item If $\bq \equiv \overline{\bp}$ then we conclude that there is no error in the suffix $\bp\bq$, consequently there is at most one error in $\bz_i$. We then use $\dec^{\rm L}_a(\bz_i)$ to correct $\bz_i$ where $a$ is the integer in  $\bbZ_{2\ell}$ whose binary representation is  $\bp$.
\end{itemize}

In conclusion, $\enc_{\cal S}^{\rm ECC}(\bx) \in {\cal S}(n,\ell,[a,b])$ and is capable of correcting at most $m$ substitution errors with the assumption that the distance between any two errors is at least $\ell$ for all $\bx \in \{0,1\}^{m\ell_1}$.
\end{proof}

For completeness, we describe the corresponding decoder as follows.
\vspace{0.05in}

\noindent{${\textbf{Decoder S}}^{\bf ECC}$}. 

{\sc Input}: $\by \in \{0,1\}^{m\ell}$\\
{\sc Output}: $\bx \triangleq \dec_{\cal S}^{\rm ECC}(\by) \in \{0,1\}^{m\ell_1}$ \\[-2mm]

\begin{enumerate}[(I)]

\item {\bf For} $1\le i\le m$ {\bf Do}:
\begin{itemize}
\item Set $\by_i={\bf B}_{(i,\ell)}(\by)$
\item Set $\bz_i$ be the prefix of length $\ell_2=\ell-2\log 2\ell$ of $\by_i$, $\bp$ be the following $\log 2\ell$ bits and $\bq$ be the suffix of length $\log 2\ell$, i.e. $\by_i=\bz_i \bp \bq$
\item {\bf If} ($\bq\equiv \overline{\bp}$) {\bf Do}:
\begin{enumerate}[(i)]
\item Let $a \in \bbZ_{2\ell}$ whose binary representation is  $\bp$
\item Let $\bc_i=\dec^{\rm L}_a(\bz_i)$ of length $\ell_2=\ell-2\log 2\ell$
\item Use Decoder S to obtain $\bx_i= \dec_{\cal S}(\bc_i)$ of length $\ell_1$
\end{enumerate}
\item {\bf If} ($\bq\neq \overline{\bp}$) {\bf Do}
\begin{enumerate}[(i)]
\item Let $\bc_i\equiv \bz_i$
\item Use Decoder S to obtain $\bx_i= \dec_{\cal S}(\bc_i)$ of length $\ell_1$
\end{enumerate}
\end{itemize}

\item Output $\bx=\bx_1\bx_2\ldots\bx_m \in \{0,1\}^{m\ell_1}$ 
\end{enumerate}
\vspace{0.05in}

\noindent{\bf Analysis.} Since Encoder S/Decoder S has linear-time encoding/decoding complexity and the error correction decoder for each subblock $\dec^{\rm L}_a(\bz_i)$ also has linear-time complexity, both Encoder ${\rm S}^{\rm ECC}$ and Decoder ${\rm S}^{\rm ECC}$ have linear-time complexity. The redundancy for error-correction in each subblock is $2\log 2\ell$. Consequently, the total redundancy for codewords of length $n=m\ell$ is then $m(r+2\log 2\ell)$. Recall that $r=\Theta(1)$. Therefore, this encoding method is efficient when the number of subblocks is small compared to the length of codeword, i.e. $m=\Theta(1),\ell=\Theta(n)$ or $m=o(n)$. In such cases, the rate of Encoder ${\rm S}^{\rm ECC}$ approach the channel capacity for sufficient $\ell, n$,
\begin{align*}
\lim_{n \to \infty} \frac{m(\ell-2\log 2\ell-r)}{m\ell} &= \lim_{\ell \to \infty} \frac{\ell-2\log 2\ell-r}{\ell} \\
&= \lim_{\ell \to \infty} 1 - \frac{\log 2\ell + r}{\ell} \\
&= 1.
\end{align*}



\subsection{Construction of SWCCs with Error-Correction Capability}

In order to combine Encoder W/Decoder W with error-correction capability, we need to make sure that after appending the syndrome to the end of the information data, any overlapping window of size $\ell$ between two parts does not violate the weight constraint. Specifically, suppose that $\bx=\bx_1\bx_2\ldots \bx_m \in {\cal W}(n,\ell,[p_1\ell,p_2\ell]$, where $\bx_i$ is of length $\ell$, and we append the balanced suffix $\by_i$ (representing the syndrome of $\bx_i$) to the end of $\bx_i$, any window of size $\ell$ in $\bx_i \by_i$ and $\by_i\bx_{i+1}$ must not be a forbidden window. The following result is crucial to the method of appending the syndrome in such a way that the weight constraint is preserved. 

For constant $p_1,p_2$ where $0\le p_1<1/2<p_2\le 1$, let $p_1'=1/2(p_1+1/2)$ and $p_2=1/2(p_2+1/2)$, and $\ell$ be sufficient that $\ell(1/2-p_1)\ge 2\log 2\ell+1$ and $\ell(p_2-1/2)\ge 2\log 2\ell+1$.

\begin{definition}
Given two binary sequences of same length $\bx=x_1x_2\ldots x_n$ and $\by=y_1y_2\ldots y_n$, the interleaved sequence of $\bx$ and $\by$ is defined by $\bx || \by \triangleq x_1 y_1 x_2 y_2 \ldots x_n y_n$. 
\end{definition}
For a binary sequence $\bx\in\{0,1\}^n$, recall that $\overline{\bx}$ denote the complement of $\bx$. Clearly, we get $\bx||\overline{\bx}$ is balanced.

\begin{lemma}\label{overlapping}
Given $\bx \in \{0,1\}^\ell$ such that ${\rm wt}(\bx)\in[p_1'\ell,p_2'\ell]$, and $\by\in\{0,1\}^m$ where $m\le \log 2\ell$. Set $\bz=\by||\overline{\by}$. For $1\le i\le 2\log 2\ell$, let $\bu_i$ be the suffix of length $(\ell-i)$ of $\bx$ and $\bv_i$ be the prefix of length $i$ of $\bz$. We then have ${\rm wt}(\bu_i\bv_i)\in[p_1\ell,p_2\ell]$ for $1\le i\le 2\log 2\ell$. 
\end{lemma}

\begin{proof}
For $1\le i\le 2\log 2\ell$, we first show that ${\rm wt}(\bu_i\bv_i)\ge p_1\ell$. Since ${\rm wt}(\bx)\in[p_1'\ell,p_2'\ell]$, and $\bu_i$ is the suffix of length $(\ell-i)$ of $\bx$, we get ${\rm wt}(\bu_i)\ge p_1'\ell-i$. On the other hand, we observe that ${\rm wt}(\bv_i)\ge (i-1)/2$. Hence, 
\begin{align*}
{\rm wt}(\bu_i\bv_i)&\ge (p_1'\ell-i)+(i-1)/2= p_1'\ell-(i+1)/2 \\
&= 1/2(p_1+1/2)\ell-(i+1)/2\\
&\ge 1/2[\underbrace{(1/2-p_1)\ell-(i+1)}_{\ge 0}]+p_1\ell \\
&\ge p_1\ell. 
\end{align*}

Similarly, we have ${\rm wt}(\bu_i)\le p_2'\ell$, ${\rm wt}(\bv_i)\le i$, and hence, 
\begin{align*}
{\rm wt}(\bu)&\le p_2'\ell+i= 1/2(p_2+1/2)\ell+i\\
&\le \underbrace{1/2(1/2-p_2)\ell+i}_{\le 0}+p_2\ell \\
&\le p_2\ell. 
\end{align*}

In conclusion, ${\rm wt}(\bu_i\bv_i)\in[p_1\ell,p_2\ell]$ for $1\le i\le 2\log 2\ell$.
\end{proof} 

\begin{corollary}\label{overlapping-strong}
Given $\bx \in \{0,1\}^\ell$ such that ${\rm wt}(\bx)\in[p_1'\ell,p_2'\ell]$, and $\by\in\{0,1\}^m$ where $m\le \log 2\ell$. Set $\bz=\by||\overline{\by}$. Let $\bv$ be any substring of length $i$ of $\bz$. Let $\bx'=\bx_1\bx_2$ be a substring of length $\ell-i$ of $\bx$ and let $\bu=\bx_1\bv\bx_2$. We then have ${\rm wt}(\bu)\in[p_1\ell,p_2\ell]$. 
\end{corollary}

\begin{proof}
Similar to the proof of Lemma~\ref{overlapping}, we can show that 
\begin{align*}
{\rm wt}(\bu) &\ge (p_1'\ell-i)+(i-1)/2 \ge p_1\ell, \text{ and } \\
{\rm wt}(\bu) &\le p_2'\ell+i \ge p_2\ell. 
\end{align*} 
Therefore, ${\rm wt}(\bu)\in[p_1\ell,p_2\ell]$.
\end{proof}

\begin{corollary}\label{add syndrome}
Let $\bx=\bx_1\bx_2 \in {\cal W}(2\ell, \ell, [p_1'\ell,p_2'\ell]$. Let $a={\rm Syn}(\bx_1) \ppmod{2\ell}$ and set $\bp$ be the binary representation of $a$ of length $\log 2\ell$. Let $\by=\bx_1(\bp||\overline{\bp})\bx_2$. There is no forbidden window in $\by$, in other words, $\by\in {\cal W}(2\ell+2\log 2\ell, \ell, [p_1\ell,p_2\ell])$.
\end{corollary}

\begin{proof}
Consider a window of size $\ell$ of $\by$. We have three following cases. 
\begin{itemize}
\item Case 1. The window includes the suffix of length $(\ell-i)$ of $\bx_1$ and a prefix of length $i$ of $\bp||\overline{\bp}$ where $i\le 2\log 2\ell$. Clearly, it is not a forbidden window, according to Lemma~\ref{overlapping}.
\item Case 2. The window is of the form $\bu\bv\bw$ where $\bu$ is the suffix of length $i$ of $\bx_1$, $\bv\equiv \bp||\overline{\bp}$ and $\bw$ is the prefix of length $\ell-i-2\log 2\ell$ of $\bx_2$. Clearly, it is not a forbidden window, according to Corollary~\ref{overlapping-strong}.
\item Case 3. The window includes the suffix of length $i$ of $\bp||\overline{\bp}$ and a prefix of length $(\ell-i)$ of $\bx_2$. Similar to case 1, it is not a forbidden window.
\end{itemize} 
In conclusion, we have $\by\in {\cal W}(2\ell+2\log 2\ell, \ell, [p_1\ell,p_2\ell])$.
\end{proof}
In the following, we present efficient encoder/decoder for SWCCs with error-correction capability. For simplicity, we assume that $n=m\ell$. Recall that $p_1'=1/2(p_1+1/2)$ and $p_2=1/2(p_2+1/2)$, and $\ell(1/2-p_1)\ge 2\log 2\ell+1$ and $\ell(p_2-1/2)\ge 2\log 2\ell+1$.
\vspace{0.05in}

\noindent{${\textbf{Encoder W}}^{\bf ECC}$}. 

{\sc Input}: $\bx \in \{0,1\}^{n-1}$\\
{\sc Output}: $\bc \triangleq \enc_{\cal W}^{\rm ECC}(\bx) \in {\cal W}(n+2m\log 2\ell,\ell,[p_1\ell,p_2\ell])$ \\[-2mm]

\begin{enumerate}[(I)]
\item Use the Encoder W to obtain $\by = \enc_{\cal W}(\bx) \in {\cal W}(n,\ell,[p_1'\ell_2,p_2'\ell_2])$. In other words, the Encoder W is constructed based on the values of $p_1',p_2'$. Suppose that $\by=\by_1\by_2\ldots \by_m$ where $\by_i\in \{0,1\}^{\ell} \cap {\cal W}(\ell,\ell,[p_1'\ell_2,p_2'\ell_2])$ for $1\le i\le m$.

\item For $1\le i\le m$ Do:
\begin{itemize}
\item Compute $a_i={\rm Syn}(\by_i) \ppmod{2\ell}$
\item Set $\bp_i$ be the binary representation of $a_i$ of length $\log 2\ell$
\item Set $\bq_i=\bp_i || \overline{\bp}$ of length $2\log 2\ell$ and $\bq_i$ is balanced
\item Set $\bc_i=\by_i \bq_i$ of length $\ell+2\log 2\ell$
\end{itemize}

\item Output $\bc=\bc_1\bc_2\ldots\bc_m$
\end{enumerate}

\begin{theorem}
The ${\text{Encoder W}}^{{\rm ECC}}$ is correct. In other words, $\enc_{\cal W}^{\rm ECC}(\bx) \in {\cal W}(n+2m\log 2\ell,\ell,[p_1\ell,p_2\ell])$, which can correct up to $m$ substitution errors with the assumption that the distance between any two errors is at least $\ell$ for all $\bx\in \{0,1\}^{n-1}$. The redundancy of ${\text{Encoder W}}^{{\rm ECC}}$ is $1+2m\log 2\ell$ (bits). 
\end{theorem}
\begin{proof}
Let $\bc=\enc_{\cal W}^{\rm ECC}(\bx)$. We first show that $\bc \in {\cal W}(n+2m\log 2\ell,\ell,[p_1\ell,p_2\ell])$, in other words, there is no forbidden window in $\bc$. Since $\by_i \in {\cal W}(\ell,\ell,[p_1'\ell,p_2'\ell]) \subset {\cal W}(\ell,\ell,[p_1\ell,p_2\ell])$, we only need to show that any window of size $\ell$ in $\by_i \bp_i \bq_i \by_{i+1}$ is not a forbidden window for $1\le i\le m-1$. This follows directly from Corollary~\ref{add syndrome}.

It remains to show that $\bc$ can correct $m$ substitution errors. To do so, we provide an efficient decoding algorithm and this algorithm is similar to the case of Encoder/Decoder ${\rm S}^{\rm error}$ as discussed in the earlier section. Suppose that we receive a sequence $\bc'=\bc_1'\bc_2'\ldots\bc_m'$ where each subblock $\bc_i'$ is of length $\ell+2\log 2\ell$. For $1\le i\le m$, we decode the $i$th subblock as follows. Let $\bz_i$ be the prefix of length $\ell$ of $\bc_i'$ and $\bq_i$ be the suffix of length $2\log 2\ell$, and $\bq_i= \bp_i ||\bp_i'$
\begin{itemize}
\item If $\bp_i'\neq \overline{\bp_i}$, then we conclude that there is an error in $\bq_i$, consequently there is no error in $\bz_i$. The decoder use Decoder W to decode $\bz_i$.
\item If $\bp_i' \equiv \overline{\bp_i}$ then we conclude that there is no error in the suffix $\bq_i$, consequently there is at most one error in $\bz_i$. We then use $\dec^{\rm L}_a(\bz_i)$ to correct $\bz_i$ where $a$ is the integer in $\bbZ_{2\ell}$ whose binary representation is $\bp_i$.
\end{itemize}

In conclusion, $\enc_{\cal W}^{\rm ECC}(\bx) \in {\cal W}(n+2m\log 2\ell,\ell,[p_1\ell,p_2\ell])$ and is capable of correcting at most $m$ substitution errors with the assumption that the distance between any two errors is at least $\ell$ for all $\bx \in \{0,1\}^{n-1}$.
\end{proof}

For completeness, we describe the corresponding decoder as follows.
\vspace{0.05in}

\noindent{${\textbf{Decoder W}}^{\bf ECC}$}. 

{\sc Input}: $\by \in \{0,1\}^{n+2m\log 2\ell}$,\\
{\sc Output}: $\bx \triangleq \dec_{\cal W}^{\rm ECC}(\by) \in \{0,1\}^{n-1}$ \\[-2mm]

\begin{enumerate}[(I)]

\item {\bf For} $1\le i\le m$ {\bf Do}:
\begin{itemize}
\item Set $\by_i={\bf B}_{(i,\ell+2\log 2\ell)}(\by)$
\item Set $\bz_i$ be the prefix of length $\ell$ of $\by_i$, $\bq_i$ be the following $2\log 2\ell$ bits and $\bq_i=\bp_i || \bp_i'$
\item {\bf If} ($\bp_i'\equiv \overline{\bp_i}$) {\bf Do}:
\begin{enumerate}[(i)]
\item Let $a \in \bbZ_{2\ell}$ whose binary representation is  $\bp_i$
\item Let $\bc_i=\dec^{\rm L}_a(\bz_i)$ of length $\ell$
\end{enumerate}
\item {\bf If} ($\bp_i'\neq \overline{\bp_i}$) {\bf Do}
\begin{enumerate}[(i)]
\item Let $\bc_i\equiv \bz_i$
\end{enumerate}
\end{itemize}
\item Let $\bc=\bc_1\bc_2\ldots \bc_m \in \{0,1\}^n \cap {\cal W}(n,\ell,[p_1'\ell,p_2'\ell])$
\item Use Decoder W to obtain $\bx=\dec_{\cal W}(\bc)$ of length $n-1$

\item Output $\bx$ 
\end{enumerate}
\vspace{0.05in}

\noindent{\bf Analysis.} Since Encoder W/Decoder W has linear-time encoding/decoding complexity and the error correction decoder for each subblock also has linear-time complexity, both Encoder ${\rm W}^{\rm ECC}$ and Decoder ${\rm W}^{\rm ECC}$ have linear-time complexity. The total redundancy of Encoder W is $1+2m\log 2\ell$, which is slightly less than the redundancy of Encoder S. This encoding method is efficient when the number of subblocks is small compared to the length of codeword, i.e. $m=o(n)$. In such cases, the rate of Encoder ${\rm W}^{\rm ECC}$ approaches the channel capacity for sufficient large $\ell, n$,
\begin{align*}
\lim_{n \to \infty} \frac{n-1}{n+2m\log 2\ell} =1.
\end{align*}

\section{Conclusion}
We have presented novel and efficient encoders that translate source binary data to codewords in SECCs, SWCCs, bounded SECCs, and bounded SWCCs. Our coding methods, based on Knuth's balancing technique and sequence replacement technique, incur low redundancy and have linear-time complexity. For certain codes parameters, our methods incur only one redundant bit. We also imposed minimum distance constraint to the designed codewords for error correction capability.

\end{document}